\documentclass[a4paper,UKenglish]{lipics-v2018}

\usepackage{microtype}%if unwanted, comment out or use option "draft"
\nolinenumbers

\renewcommand{\log}{\lg}

\def\idtt#1{\ensuremath{\mathtt{#1}}}

\def\parent{\idtt{parent}}
\def\nodedepth{\idtt{depth}}
\def\treeroot{\idtt{root}}
\def\nodelabel{\idtt{label}}
\def\prevlabel{\idtt{prevlabel}}
\def\lca{\idtt{lca}}
\def\anc{\idtt{anc}}
\def\labelanc{\idtt{labelanc}}
\def\pathcount{\idtt{count}}
\def\preorder{\idtt{preorder}}
\def\postorder{\idtt{postorder}}
\def\subtreesize{\idtt{subtreesize}}
\def\prank{\idtt{prank}}
\def\rank{\idtt{rank}}
\def\access{\idtt{access}}
\def\select{\idtt{select}}
\def\open{\idtt{open}}
\def\close{\idtt{close}}
\def\enclose{\idtt{enclose}}

\def\etal{{\em et al.}}

\title{Tree Path Majority Data Structures}

\author{Travis Gagie}{CeBiB --- Center for Biotechnology and Bioengineering, Chile; 
School of Computer Science and Telecommunications, Diego Portales University, Chile}{travis.gagie@gmail.com}{}{Funded by FONDECYT grant 1171058, Chile.}
\author{Meng He}{Faculty of Computer Science, Dalhousie University, Canada}{mhe@cs.dal.ca}{}{Funded by NSERC, Canada.}
\author{Gonzalo Navarro}{CeBiB --- Center for Biotechnology and Bioengineering, Chile;
IMFD --- Millenium Institute for Foundational Research on Data, Chile;
Department of Computer Science, University of Chile,
Chile}{gnavarro@dcc.uchile.cl}{}{Funded with basal funds FB0001, Conicyt,
Chile, by Millenium Institute for Foundational Research on Data (IMFD), Chile,
and by Fondecyt grant 1170048, Chile.}

\authorrunning{T.~Gagie, M.~He and G.~Navarro}

\Copyright{Travis Gagie, Meng He, Gonzalo Navarro}



\begin{document}

\maketitle

\begin{abstract}
We present the first solution to $\tau$-majorities on tree paths. Given a
tree of $n$ nodes, each with a label from $[1..\sigma]$, and a fixed threshold
$0<\tau<1$, such a query gives two nodes $u$ and $v$ and asks for all the
labels that appear more than $\tau \cdot |P_{uv}|$ times in the path 
$P_{uv}$ from $u$ to $v$, where $|P_{uv}|$ denotes the number of nodes in $P_{uv}$.
Note that the answer to any query is of size up to 
$1/\tau$. 
On a $w$-bit RAM,
we obtain a linear-space data structure 
%that lists all the majorities in time%
with $O((1/\tau)\log^* n \log\log_w \sigma)$ query time.
For any $\kappa > 1$, we can also build a structure that uses
$O(n\log^{[\kappa]} n)$ space, where $\log^{[\kappa]} n$ denotes the function that applies logarithm $\kappa$ times to $n$,
and answers queries in time $O((1/\tau)\log\log_w \sigma)$. 
The construction time of both structures is $O(n\log n)$.
We also describe two succinct-space solutions with the same query time
of the linear-space structure. One uses $2nH + 4n + o(n)(H+1)$ bits,
where $H \le \lg\sigma$ is the entropy of the label distribution,
and can be built in $O(n\log n)$ time. The other uses $nH + O(n) + o(nH)$ bits
and is built in $O(n\log n)$ time w.h.p.
\end{abstract}

\section{Introduction}
\label{sec:intro}

Finding frequent elements in subsets of a multiset is a fundamental 
operation for data analysis and data mining \cite{fsmmu1998,dlm2002}. When the 
sets have a certain structure, it is possible to preprocess the multiset to 
build data structures that efficiently find the frequent elements in any subset.

The best studied multiset structure is the sequence, where the subsets that
can be queried are ranges (i.e., contiguous subsequences) of the sequence. 
Applications of this case include time sequences, linear-versioned structures, 
and one-dimensional models, for example.
Data structures for finding the {\em mode} (i.e., the most frequent element) 
in a range require time $O(\sqrt{n/\log n})$, and it is unlikely that this can 
be done much better within reasonable extra space \cite{cdlmw14}. Instead, 
listing all the elements whose relative frequency in a range is over some
fraction $\tau$ (called the {\em $\tau$-majorities} of the range)
is feasible within linear space and $O(1/\tau)$ time, which is worst-case
optimal \cite{bgmnn2016}. 
Mode and $\tau$-majority queries on higher-dimensional arrays have also been
studied \cite{ghmn2011,cdlmw14}.

In this paper we focus on finding frequent elements when the subsets that can 
be queried are the labels on paths from one given node to another in a labeled 
tree. For example, given a minimum spanning tree of a graph, we might be 
interested in frequent node types on the path between two nodes. Path mode or 
$\tau$-majority queries on multi-labeled trees could be useful when handling
the tree of versions of a document or a piece of software, or a phylogenetic tree (which is essentially a tree of versions of a genome).  If each node stores a list of the sections (i.e., chapters, modules, genes) on which its version differs from its parent's,
then we can efficiently query which sections are changed most frequently between two given versions.

There has been relatively little work previously on finding frequent elements on tree paths.
Krizanc~\etal~\cite{KMS05} considered path mode queries, obtaining 
$O(\sqrt{n}\log n)$ query time. This was recently improved by Durocher~\etal~\cite{DSST16},
who obtained $O(\sqrt{n/w}\log\log n)$ time on a RAM machine of 
$w=\Omega(\log n)$ bits. Like on the more special case of sequences, these
times are not likely to improve much. No previous work has considered the 
problem of finding path $\tau$-majority queries, which is more tractable
than finding the path mode. This is our focus.

We present the first data structures to support path
$\tau$-majority queries on trees of $n$ nodes, with labels in $[1..\sigma]$, 
on a RAM machine. We first obtain a data structure using $O(n\log n)$ space and
$O((1/\tau)\log\log_w\sigma)$ time (Theorem~\ref{thm:nlogn}). Building on
this result, we reduce the space to $O(n)$ at the price of 
a very slight increase in the query time, $O((1/\tau)\log^*n\log\log_w\sigma)$
(Theorem~\ref{thm:linear}). We then show that the original query time can be
obtained within very slightly superlinear space, $O(n\log^{[\kappa]} n)$ for
any desired $\kappa>1$, where $\log^{[\kappa]} n$ denotes the function that applies logarithm $\kappa$ times to $n$ 
(Theorem~\ref{thm:superlinear}). Finally, we show that our linear-space data
structure can be further compressed, to either $2nH + 4n + o(n)(H+1)$ bits or
$nH + O(n) + o(nH)$ bits, where $H \le \lg \sigma$ is the entropy of the
distribution of the labels in $T$, while retaining the same query times of 
the linear-space data structure (Theorems~\ref{thm:succ1} and \ref{thm:succ2}).
All our structures can be built in $O(n\log n)$ deterministic time; only the
latter one requires that time only w.h.p.  We close with a brief
discussion of directions for future research.  In particular, we describe
how to adapt our results to multi-labeled trees.

Durocher~\etal~\cite{DSST16} also considered
queries that look for the least frequent elements and 
$\tau$-minorities on paths. In the appendix, we show how to compress their 
data structure for $\tau$-minorities with only a very slight increase in query 
time. 

\section{Preliminaries}
\label{sec:pre}

\subsection{Definitions}

We deal with {\em rooted ordinal trees} (or just {\em trees}) $T$. Further,
our trees are {\em labeled}, that is, each node $u$ of $T$ has an
integer label $\nodelabel(u) \in [1..\sigma]$. We assume that, if our main tree
has $n$ nodes, then $\sigma = O(n)$ (we can always remap the labels to a range
of size at most $n$ without altering the semantics of the queries of interest
in this paper).

The {\em path} between nodes $u$ and $v$ in a tree $T$ is the (only) sequence
of nodes $P_{uv} = \langle u=z_1,z_2,\ldots,z_{k-1},z_k=v\rangle$ such that 
there is an edge in $T$ between each pair $z_i$ and $z_{i+1}$, for $1\le i<k$.
The length of the path is $|P_{uv}| = k$, for example the length of the path 
$P_{uu}$ is 1. Any path from $u$ to $v$ goes from $u$ to the lowest common
ancestor of $u$ and $v$, and then from there it goes to $v$ (if $u$ is an
ancestor of $v$ or vice versa, one of these two subpaths is empty).

Given a real number $0 < \tau < 1$, a {\em $\tau$-majority} of the
path $P_{uv}$ is any label that appears (strictly) more than $\tau \cdot
|P_{uv}|$ times among the labels of the nodes in $P_{uv}$. The {\em path
$\tau$-majority problem} is, given $u$ and $v$, list all the
$\tau$-majorities in the path $P_{uv}$. Note that there can be up to
$\lfloor 1/\tau \rfloor$ such $\tau$-majorities.

Our results hold in the RAM model of computation, assuming a computer word of
$w=\Omega(\log n)$ bits, supporting the standard operations.

Our logarithms are to the base 2 by default.
By $\log^{[k]} n$ we mean the function that applies logarithm $k$ times to $n$, i.e., $\log^{[0]} n = n$ and
$\log^{[k]} n = \log (\log^{[k-1]} n)$. By $\log^* n$ we denote the iterated logarithm, i.e., the minimum
$k$ such that $\log^{[k]} n \le 1$.

\subsection{Sequence representations}

A bitvector $B[1..n]$ can be represented within $n+o(n)$ bits so that the
following operations take constant time: $\access(B,i)$ returns $B[i]$,
$\rank_b(B,i)$ returns the number of times bit $b$ appears in $B[1..i]$,
and $\select_b(B,j)$ returns the position of the $j$th occurrence of $b$
in $B$ \cite{Cla96}. If $B$ has $m$ 1s, then it can be represented within
$m\log(n/m)+O(m)$ bits while retaining the same operation times \cite{RRR07}.
Those structures can be built in linear time. Note the space is $o(n)$ bits
if $m=o(n)$.

Analogous operations are defined on sequences $S[1..n]$ over alphabets
$[1..\sigma]$. For example, one can represent $S$ within $nH + o(n)(H+1)$
bits, where $H\le\lg\sigma$ is the entropy of the distribution of symbols in 
$S$, so that $\rank$ takes time $O(\log\log_w \sigma)$, $\access$ takes time
$O(1)$, and $\select$ takes any time in $\omega(1)$ \cite[Thm.~8]{BN15}. The
construction takes linear time. While this $\rank$ time is optimal,
we can answer {\em partial rank} queries in $O(1)$ time, 
$\prank(S,i)=\rank_{S[i]}(S,i)$,
by adding $O(n(1+\log H))$ bits on top of a representation giving
constant-time $\access$ \cite[Sec.~3]{BN13}. This construction requires linear
randomized time.

\subsection{Range $\tau$-majorities on sequences}
\label{sec:range}

A special version of the path $\tau$-majority queries on trees
is range $\tau$-majority queries on sequences $S[1..n]$, which are
much better studied. Given $i$ and $j$, the problem is to return all the distinct
symbols that appear more than $\tau \cdot (j-i+1)$ times in $S[i..j]$.
The most recent result on this problem \cite{BGN13,bgmnn2016} is a linear-space
data structure, built in $O(n\log n)$ time, that answers queries in the 
worst-case optimal time, $O(1/\tau)$.

For our succinct representations, we also use a data structure
\cite[Thm.~6]{bgmnn2016} that requires $nH + o(n)(H+1)$ bits, and can answer
range $\tau$-majority queries in any time in $(1/\tau)\cdot\omega(1)$. The 
structure is built on the sequence representation mentioned above 
\cite[Thm.~8]{BN15}, and thus it includes its support for $\access$, $\rank$, 
and $\select$ queries on the sequence. To obtain the given times for
$\tau$-majorities, the structure includes the support for partial rank queries
\cite[Sec.~3]{BN13}, and therefore its construction time is randomized. In this
paper, however, it will be sufficient to obtain $O((1/\tau)\log\log_w\sigma)$
time, and therefore we can replace their $\prank$ queries by general $\rank$
operations. These take time $O(\log\log_w\sigma)$ instead of $O(1)$, but can be
built in linear time.\footnote{In fact, their structure \cite{bgmnn2016} can
be considerably simplified if one can spend the time of a general $\rank$
query per returned majority.}
Therefore, this slightly slower structure can also be built in $O(n\log n)$ 
deterministic time.

When a set has no structure, we can find its $\tau$-majorities in linear time.
Misra and Gries~\cite{mg1982} proposed an optimal solution that computes all $\tau$-majorities using $O(n \lg (1/\tau))$ comparisons. When implemented on a word RAM over an integer alphabet of size $\sigma$, the running time becomes 
$O(n)$~\cite{dlm2002}.

\subsection{Tree operations}
\label{sec:treeops}

For tree nodes $u$ and $v$, we define the operations 
$\treeroot$ (the tree root), 
$\parent(u)$ (the parent of node $u$), 
$\nodedepth(u)$ (the depth of node $u$, 0 being the depth of the root),
$\preorder(u)$ (the rank of $u$ in a preorder traversal of $T$),
$\postorder(u)$ (the rank of $u$ in a postorder traversal of $T$),
$\subtreesize(u)$ (the number of nodes descending from $u$, including $u$),
$\anc(u,d)$ (the ancestor of $u$ at depth $d$), and
$\lca(u,v)$ (the lowest common ancestor of $u$ and $v$).
All those operations can be supported in constant time and linear space on a 
static tree after a linear-time preprocessing, trivially with 
the exceptions of $\anc$ \cite{BF04} and $\lca$ \cite{BFCPSS05}.

A less classical query is $\labelanc(u,\ell)$, which returns the nearest
ancestor of $u$ (possibly $u$ itself) labeled $\ell$ (note that the label of 
$u$ needs not be $\ell$). If $u$ has no ancestor labeled $\ell$,
$\labelanc(u,\ell)$ returns $null$. This operation can be solved in time
$O(\log\log_w\sigma)$ using linear space and preprocessing time
\cite{HMZ14,Tsu14,DSST16}.

\subsection{Succinct tree representations}

A tree $T$ of $n$ nodes can be
represented as a sequence $P[1..2n]$ of parentheses (i.e., a bit sequence).
In particular, we consider the balanced parentheses representation, where we
traverse $T$ in depth-first order, writing an opening parenthesis when 
reaching a node and a closing one when leaving its subtree. 
A node is identified with
the position $P[i]$ of its opening parenthesis. By using $2n+o(n)$ bits,
all the tree operations defined in Section~\ref{sec:treeops} (except those on labels) can be supported in
constant time \cite{NS14}.

This representation also supports $\access$, $\rank$ and $\select$ on the
bitvector of parentheses, and the operations 
$\close(P,i)$ (the position of the parenthesis closing the one
that opens at $P[i]$),
$\open(P,i)$ (the position of the parenthesis opening the one
that closes at $P[i]$), and
$\enclose(P,i)$ (the position of the rightmost opening parenthesis whose corresponding parenthesis pair encloses 
$P[i]$; when $P$ represents a tree, this parenthesis represents the parent of the node that $P[i]$ corresponds to).

Labeled trees can be represented within $nH + 2n + o(n)(H+1)$ bits by adding
the sequence $S[1..n]$ of the node labels in preorder, so that $\nodelabel(i)
= \access(S,\preorder(i))$.

\section{An $O(n\log n)$-Space Solution}
\label{sec:nlogn}

In this section we design a data structure answering path $\tau$-majority 
queries on a tree of $n$ nodes using $O(n\log n)$ space and 
$O((1/\tau)\log\log_w\sigma)$ time. This is the basis to obtain our final 
results.

We start by {\em marking} $O(\tau n)$ tree nodes, in a way that any node
has a marked ancestor at distance $O(1/\tau)$. A simple way to obtain
these bounds is to mark every node whose height is $\ge \lceil 1/\tau\rceil$
and whose depth is a multiple of $\lceil 1/\tau \rceil$.
Therefore, every marked node is the nearest marked ancestor of 
at least $\lceil 1/\tau\rceil -1$ distinct non-marked nodes, which 
guarantees that there are $\le \tau n$ marked nodes. On the other hand, any
node is at distance at most $2\lceil 1/\tau\rceil-1$ from its nearest marked 
ancestor.

For each marked node $x$, we will consider prefixes $P_i(x)$ of the labels in
the path from $x$ to the root, of length $1+2^i$, that is, 
$$P_i(x)=\langle \nodelabel(x),\nodelabel(\parent(x)),\nodelabel(\parent^2(x)), 
\ldots, \nodelabel(\parent^{2^i}(x)) \rangle$$
(terminating the sequence at the root if we reach it). For each 
$0 \le i \le \lceil\lg\nodedepth(x)\rceil$, we
store $C_i(x)$, the set of $(\tau/2)$-majorities in $P_i(x)$.
Note that $|C_i(x)| \le 2/\tau$ for any $x$ and $i$.

By successive applications of the next lemma we have that, to find all the 
$\tau$-majorities in the path from $u$ to $v$, we can partition the path 
into several subpaths and then consider just the $\tau$-majorities in each
subpath.

\begin{lemma} \label{lem:split}
Let $u$ and $v$ be two tree nodes, and let $z$ be an intermediate node in the
path. Then, a $\tau$-majority in the path from $u$ to $v$ is a
$\tau$-majority in the path from $u$ to $z$ (including $z$) or a
$\tau$-majority in the path from $z$ to $v$ (excluding $z$), or in both.
\end{lemma}
\begin{proof}
Let $d_{uz}$ be the distance from $u$ to $z$ (counting $z$) and $d_{zv}$
be the distance from $z$ to $v$ (not counting $z$). Then the path from $u$ to 
$v$ is of length $d=d_{uz}+d_{zv}$. If a label $\ell$ occurs at most 
$\tau \cdot d_{uz}$ times in the path from $u$ to $z$ and at most 
$\tau \cdot d_{zv}$ times in the path from $z$ to $v$, then it occurs at most 
$\tau (d_{uz}+d_{zv}) = \tau \cdot d$ times in the path from $u$ to $v$.
\end{proof}

Let us now show that the candidates we record for marked nodes are sufficient
to find path $\tau$-majorities towards their ancestors.

\begin{lemma} \label{lem:cand}
Let $x$ be a marked node. All the $\tau$-majorities in the path from $x$ to 
a proper ancestor $z$ are included in $C_i(x)$ for some suitable $i$.
\end{lemma}
\begin{proof}
Let $d_{xz}=\nodedepth(x)-\nodedepth(z)$ be the distance from $x$ to $z$ (i.e., the
length of the path from $x$ to $z$ minus 1).
Let $i = \lceil \log d_{xz} \rceil$. The path $P_i(x)$ contains all the nodes 
in an upward path of length $1+2^i$ starting at $x$,
where $d_{xz} \le 2^i < 2 d_{xz}$. Therefore, $P_i(x)$ contains node $z$, but 
its length is $|P_i(x)| < 1+2d_{xz}$. Therefore, any $\tau$-majority in the 
path from $x$ to $z$ appears more than $\tau\cdot (1+d_{xz}) > 
(\tau/2)\cdot(1+2d_{xz}) > (\tau/2)\cdot |P_i(x)|$ times, and thus it is 
a $(\tau/2)$-majority recorded in $C_i(x)$.
\end{proof}

\subsection{Queries}

With the properties above, we can find a candidate set of size 
$O(1/\tau)$ for the path $\tau$-majority between arbitrary tree nodes $u$
and $v$. Let $z=\lca(u,v)$. If $v \not= z$, let us also define
$z'=\anc(v,\nodedepth(z)+1)$, that is, the child of $z$ in the path to $v$. 
The path is then split into at most four subpaths, 
each of which can be empty:
\begin{enumerate}
\item The nodes from $u$ to its nearest marked ancestor, $x$, not including
$x$. If $x$ does not exist or is a proper ancestor of $z$, then this subpath 
contains the nodes from $u$ to $z$. The length of this path is less than $2\lceil 1/\tau
\rceil$ by the definition of marked nodes, and it is empty if $u=x$.
\item The nodes from $v$ to its nearest marked ancestor, $y$, not including
$y$. If $y$ does not exist or is an ancestor of $z$, then this subpath
contains the nodes from $v$ to $z'$. The length of this path is again less than $2\lceil 1/\tau
\rceil$, %by the definition of marked nodes,
and it is empty if $v=y$ or $v=z$.
\item The nodes from $x$ to $z$. This path exists only if $x$ exists and
descends from $z$.
\item The nodes from $y$ to $z'$. This path exists only if $y$ exists and
descends from $z'$.
\end{enumerate}

By Lemma~\ref{lem:split}, any $\tau$-majority in the path from $u$ to $v$
must be a $\tau$-majority in some of these four paths. For the paths 1 and
2, we consider all their up to $2\lceil 1/\tau\rceil-1$ nodes as candidates. 
For the paths 3 and
4, we use Lemma~\ref{lem:cand} to find suitable values $i$ and $j$ so that
$C_i(x)$ and $C_j(y)$, both of size at most $2/\tau$, contain all the possible
$\tau$-majorities in those paths. In total, we obtain a set of at most
$8/\tau+O(1)$ candidates that contain all the $\tau$-majorities in the
path from $u$ to $v$. 

In order to verify whether a candidate is indeed a $\tau$-majority, we follow the technique of Durocher et 
al.~\cite{DSST16}. Every tree node $u$ will store $\pathcount(u)$, the number 
of times its label occurs in the path from $u$ to the root. We also make use 
of operation $\labelanc(u,\ell)$. If $u$ has no ancestor labeled $\ell$, this 
operation returns $null$, and we define $\pathcount(null)=0$. 
Therefore, the number of times label $\ell$ occurs in the path from $u$ to an 
ancestor $z$ of $u$ (including $z$) can be computed as 
$\pathcount(\labelanc(u,\ell))-\pathcount(\labelanc(\parent(z),\ell))$. 
Each of our candidates can then be checked by counting their occurrences in 
the path from $u$ to $v$ using 
\begin{eqnarray*}
&& (\pathcount(\labelanc(u,\ell))-\pathcount(\labelanc(\parent(z),\ell))) \\
&& +~ (\pathcount(\labelanc(v,\ell))-\pathcount(\labelanc(z,\ell))).
\end{eqnarray*}
The time to perform query $\labelanc$ is $O(\log\log_w\sigma)$ using a
linear-space data structure on the tree \cite{HMZ14,Tsu14,DSST16}, and 
therefore we find
all the path $\tau$-majorities in time $O((1/\tau)\log\log_w \sigma)$.

The space of our data structure is dominated by the $O(\log n)$ candidate
sets $C_i(x)$ we store for the marked nodes $x$. These amount to $O((1/\tau)
\log n)$ space per marked node, of which there are $O(\tau n)$. Thus,
we spend $O(n\log n)$ space in total.

\begin{theorem} \label{thm:nlogn}
Let $T$ be a tree of $n$ nodes with labels in $[1..\sigma]$, and $0<\tau<1$. 
On a RAM machine of $w$-bit words, we can build an $O(n\log n)$ space data 
structure that answers path $\tau$-majority queries in time 
$O((1/\tau)\log\log_w\sigma)$.
\end{theorem}

\subsection{Construction}
\label{sec:constr}

The construction of the data structure is easily carried out in linear time
(including the fields $\pathcount$ and the data structure to support
$\labelanc$ \cite{DSST16}), except for the candidate sets $C_i(x)$ of the marked
nodes $x$. We can compute the sets $C_i(x)$ for all $i$ in total time
$O(\nodedepth(x))$ using the linear-time algorithm of Misra and Gries 
\cite{mg1982} because we compute $(\tau/2)$-majorities of doubling-length 
prefixes $P_i(x)$. This amounts to time $O(mt)$ on a tree of $t$ nodes and 
% we also have to build the candidate sets, but this is included in the time
% of MG. Note we are not even using that one seq is a prefix of the next, just
% that they are doubling size.
$m$ marked nodes. In our case, where $t=n$ and $m \le \tau n$, this is 
$O(\tau n^2)$.

To reduce this time, we proceed as follows. First we build all the data 
structure components except the sets $C_i(x)$. We then decompose the tree into
heavy paths \cite{ST83} in linear time, and collect the labels along the 
heavy paths to form a set of sequences. On the sequences, we build in 
$O(t\log t)$ time the range $\tau$-majority data structure 
\cite{BGN13,bgmnn2016} that answers queries in time $O(1/\tau)$. The prefix
$P_i(x)$ for any marked node $x$ then spans $O(\log t)$ sequence ranges,
corresponding to the heavy paths intersected by $P_i(x)$. We can then
compute $C_i(x)$ by collecting and checking the $O(1/\tau)$ 
$(\tau/2)$-majorities from each of those $O(\log t)$ ranges. 

Let the path from $x$ to the root be formed by $O(\log t)$ heavy path 
segments $\pi_1, \ldots, \pi_k$ We first compute the $O(1/\tau)$ 
$(\tau/2)$-majority in the sequences corresponding to each prefix
$\pi_1,\ldots,\pi_k$: For each $\pi_j$, we (1) compute its $2/\tau$ 
majorities on the corresponding sequence in time $O(1/\tau)$, 
(2) add them to the set of $2/\tau$ majorities already computed for 
$\pi_1,\ldots,\pi_{j-1}$, and $(3)$ check the exact frequencies of all 
the $4/\tau$ candidates in the path $\pi_1,\ldots,\pi_j$ in time 
$O((1/\tau)\log\log_w \sigma)$, using the structures already computed 
on the tree. All the $(\tau/2)$-majorities for $\pi_1,\ldots,\pi_j$ are 
then found. 

Each path $P_i(x)$ is formed by some prefix $\pi_1,\ldots,\pi_j$ plus a 
prefix of $\pi_{j+1}$. We can then carry out a process similar to the one 
to compute the majorities of $\pi_1,\ldots,\pi_{j+1}$, but using only the 
proper prefix of $\pi_{j+1}$. The $O(\log t)$ sets $C_i(x)$ are then 
computed in total time
$O((1/\tau)\log t \log\log_w \sigma)$. Added over the $m$ marked nodes,
we obtain $O((1/\tau)m\log t\log\log_w\sigma)$ construction time. 

\begin{lemma} \label{lem:constr}
On a tree of $t$ nodes, $m$ of which are marked, all the candidate sets 
$C_i(x)$ can be built in time $O((1/\tau)m\log t\log\log_w\sigma)$.
\end{lemma}

The construction time in our case, where $t=n$ and $m \le \tau n$, is the 
following.

\begin{corollary}  \label{cor:constr}
The data structure of Theorem~\ref{thm:nlogn} can be built in time
$O(n\log n\log\log_w\sigma)$.
\end{corollary}

\section{A Linear-Space (and a Near-Linear-Space) Solution}
\label{sec:linear}

We can reduce the space of our data structure by stratifying our tree.
First, let us create a separate structure to handle unary paths, that is,
formed by nodes with only one child. The labels of upward maximal unary 
paths are laid out in a sequence, and the sequences of the labels of all the 
unary paths in $T$ are concatenated into a single sequence, $S$, of length 
at most $n$. On this sequence we build the linear-space data structure that 
solves range $\tau$-majority queries in time $O(1/\tau)$
\cite{BGN13,bgmnn2016}. 
Each node in a unary path of $T$ points to its position in $S$.
Each node also stores a pointer to its nearest branching ancestor
(i.e., ancestor with more than one child).

The stratification then proceeds as follows. We say that a tree node is {\em 
large} if it has more than $(1/\tau)\log n$ descendant nodes; other nodes are
{\em small}. Then the subset of the large nodes, which is closed by $\parent$,
induces a subtree $T'$ of $T$ with the same root and containing at most
$\tau n/\log n$ leaves, because for each leaf in $T'$ there are at least 
$(1/\tau)\log n-1$ distinct nodes of $T$ not in $T'$. Further, $T - T'$ 
is a forest of trees $\{ F_i \}$, each of size at most $(1/\tau)\log n$.

We will use for $T'$ a structure similar to the one of
Section~\ref{sec:nlogn}, with some changes to ensure linear space. Note that
$T'$ may have $\Theta(n)$ nodes, but since it has at most $\tau n/\log n$ 
leaves, $T'$ has only $O(\tau n/\log n)$ branching nodes.
We modify the marking scheme, so that we mark exactly 
the branching nodes in $T'$. Spending $O((1/\tau)\log n)$ space of the 
candidate sets $C_i(x)$ over all branching nodes of $T'$ adds up to $O(n)$ 
space.

The procedure to solve path $\tau$-majority queries on $T'$ is then as
follows. We split the path from $u$ to $v$ into four subpaths, exactly as in
Section~\ref{sec:nlogn}. The subpaths of type 1 and 2 can now be of arbitrary
length, but they are unary, thus we obtain their $1/\tau$ candidates in time 
$O(1/\tau)$ from the corresponding range of $S$. Finally, we check all the 
$O(1/\tau)$ candidates in time $O((1/\tau)\log\log_w \sigma)$ as in
Section~\ref{sec:nlogn}.

The nodes $u$ and $v$ may, however, belong to some small tree $F_i$, which is
of size $O((1/\tau)\log n)$. We preprocess all those $F_i$ in a way analogous
to $T$. From each $F_i$ we define $F_i'$ as the subtree of $F_i$ induced by the
($\parent$-closed) set of the nodes with more than $(1/\tau)\log\log n$
descendants; thus $F_i'$ has $O(|F_i|\tau/\log\log n)$ branching nodes, which
are marked. We store the candidate sets $C_i(x)$ of their marked nodes $x$, 
considering only the nodes in $F_i'$. 

If the candidates were stored as in Section~\ref{sec:nlogn}, they would require
$O((1/\tau)\log\sigma)$ bits per marked node. Instead of storing the 
candidate labels $\ell$ directly, however, we will store $\nodedepth(y)$, 
where $y$ is the nearest ancestor of $x$ with label $\ell$. We can then recover
$\ell=\nodelabel(\anc(x,\nodedepth(y)))$ in constant time. Since the depths in
$F_i$ are also $O((1/\tau)\log n)$, we need only
$O(\log((1/\tau)\log n))$ bits per candidate. Further, by sorting the
candidates by their $\nodedepth(y)$ value, we can encode only the differences
between consecutive depths using $\gamma$-codes \cite{BCW90}. Encoding $k$
increasing numbers in $[1..m]$ with this method requires $O(k\log(m/k))$ bits;
therefore we can encode our $O(1/\tau)$ candidates using $O((1/\tau)\log\log
n)$ bits in total. Added over all the $O(\log n)$ values of $i$,%
\footnote{The values of $i$ are also bounded by $O(\log((1/\tau)\log n))$,
but the bound $O(\log n)=O(w)$ is more useful this time.}
the candidates $C_i(x)$ require $O((1/\tau)\log\log n)$ words per marked 
(i.e., branching)
node. Added over all the branching nodes of $F_i'$, this amounts
to $O(|F_i'|)$ space. The other pointers of $F_i$, as well as node labels, can 
be represented normally, as they are $O(n)$ in total.

The small nodes left out from the trees $F_i$ form a forest of subtrees
of size $O((1/\tau)\log\log n)$ each. We can iterate this process $\kappa$ 
times, so that the smallest trees are of size $O((1/\tau)\log^{[\kappa]} n)$.
We build no candidates sets on the smallest trees.
We say that $T'$ is a subtree of {\em level} $1$, our $F_i'$ are 
subtrees of level $2$, and so on, until the smallest subtrees, which are of 
level $\kappa$. Every node in $T$ has a pointer to the root of the subtree 
where it belongs in the stratification.

The general process to solve a path $\tau$-majority query from $u$ to $v$ is
then as follows. We compute $z=\lca(u,v)$ and split the path from $u$ to $z$
into $k-k'+1$ subpaths, where $k'$ and $k$ (note $k' \le k \le \kappa$) are the levels of the 
subtree where $z$ and $u$ belong, respectively. Let us call $u_i$ the root of 
the subtree of level $i$ that is an ancestor of $u$, except that we call
$u_{k'}=z$.
\begin{enumerate}
\item If $k=\kappa$, then $u$ belongs to one of the smallest subtrees. We then
collect the $O((1/\tau)\log^{[\kappa]} n)$ node weights in the path from $u$ 
to $u_\kappa$ one by one and include them in the set of candidates. Then we move to the parent of that root,
setting $u \leftarrow \parent(u_\kappa)$ and $k\leftarrow \kappa-1$.
\item At levels $k' \le k<\kappa$, if $u$ is a branching node, we collect the
$2/\tau$ candidates from the corresponding set $C_i(u)$, where $i$ is 
sufficient to cover $u_k$ ($C_i(u)$ will not store candidates beyond the 
subtree root). We then set $u \leftarrow \parent(u_k)$ and $k\leftarrow k-1$.
\item At levels $k' \le k<\kappa$, if $u$ is not a branching node, let $x$ be
lowest between $\parent(z)$ and the nearest branching ancestor of $u$. Let also
$p$ be the position of $u$ in $S$. Then we find the $1/\tau$ 
$\tau$-majorities in $S[p..p+\nodedepth(u)-\nodedepth(x)-1]$ in time 
$O(1/\tau)$. We then continue from $u \leftarrow x$ and $k \leftarrow k(x)$,
where $k(x)$ is the level of the subtree where $x$ belongs. Note that $k(x)$
can be equal to $k$, but it can also be any other level $k' \le k(x) < k$.
\item We stop when $u=\parent(z)$.
\end{enumerate}

A similar procedure is followed to collect the candidates from $v$ to $z'$.
In total, since each path has at most one case 2 and one case 3 per level $k$,
we collect at most $4\kappa$ candidate sets of size $O(1/\tau)$,
plus two of size $O((1/\tau)\log^{[\kappa]} n)$. The total cost to verify
all the candidates is then $O((1/\tau)(\kappa+\log^{[\kappa]}
n)\log\log_w\sigma)$. The data structure uses linear space for any choice of
$\kappa$, whereas the optimal time is obtained by setting $\kappa = \log^* n$.

The construction time, using the technique of Lemma~\ref{lem:constr} in level
1, is $O(n\log\log_w\sigma)$, since $T'$ has $t=O(n)$ nodes and $m=
O(\tau n/\log n)$ marked nodes.
For higher levels, we use the basic quadratic method described in the
first lines of Section~\ref{sec:constr}: a subtree $F$ of level $k$ has
$t = O((1/\tau)\log^{[k-1]} n)$ nodes and $m=O(\tau t/\log^{[k]} n)$ 
marked nodes, so it is built in time $O(mt)$. There are 
$O(\tau n/\log^{[k-1]} n)$ trees of level $k$, which
gives a total construction time of $O(n \log^{[k-1]} n / \log^{[k]} n)$ for
all the nodes in level $k$. Added over all the levels $k>1$, this yields
$O(n\log n/\log\log n)$. Both times, for $k=1$ and $k>1$, are however dominated 
by the $O(n\log n)$ time to build the range majority data structure on $S$.

\begin{theorem} \label{thm:linear}
Let $T$ be a tree of $n$ nodes with labels in $[1..\sigma]$, and $0<\tau<1$.
On a RAM machine of $w$-bit words, we can build in $O(n\log n)$ time an $O(n)$ 
space data structure that answers path $\tau$-majority queries in time
$O((1/\tau)\log^* n \log\log_w\sigma)$. 
\end{theorem}

On the other hand, we can use any constant number $\kappa$ of levels, and
build the data structure of Section~\ref{sec:nlogn} on the last one, so as to
ensure query time $O(1/\tau)$ in this level as well. We use, however, the 
compressed storage of the candidates used in this section. With 
this storage format, a candidate set $C_i(x)$ takes 
$O((1/\tau)\log^{[\kappa]} n)$ bits. Multiplying by $\log n$ (the crude upper
bound on the number of $i$ values), this becomes $O((1/\tau)\log^{[\kappa]}n)$
words. Since the trees are of size $O((1/\tau)\log^{[\kappa-1]}n)$
and the sampling rate used in Section~\ref{sec:nlogn} is $\tau$, this amounts
to $O((1/\tau)\log^{[\kappa-1]}n\log^{[\kappa]} n)$ space per tree. 
Multiplied by the $O(\tau n/\log^{[\kappa-1]} n)$ trees of level $\kappa$, 
the total space is $O(n \log^{[\kappa]} n)$. 

The construction time of the candidate sets in the last level, using the basic 
quadratic construction, is $O(mt)=O((1/\tau)(\log^{[\kappa-1]}n)^2)$,
because $t=O((1/\tau)\log^{[\kappa-1]}n)$ and $m=\tau t$ according to the
sampling used in Section~\ref{sec:nlogn}. Multiplying by the
$O(\tau n/\log^{[\kappa-1]} n)$ trees of level $\kappa$, the total 
construction time for this last level is $O(n \log^{[\kappa-1]} n)$, again 
dominated by the time to build the range majority data structures if
$\kappa > 1$. This yields the following result.

\begin{theorem} \label{thm:superlinear}
Let $T$ be a tree of $n$ nodes with labels in $[1..\sigma]$, and $0<\tau<1$.
On a RAM machine of $w$-bit words, for any constant $\kappa > 1$, we can 
build in $O(n\log n)$ time an $O(n\log^{[\kappa]} n)$ space data structure that
answers path $\tau$-majority queries in time $O((1/\tau) \log\log_w\sigma)$.
\end{theorem}

\section{A Succinct Space Solution}

The way to obtain a succinct space structure from Theorem~\ref{thm:linear} is
to increase the thresholds that define the large nodes in
Section~\ref{sec:linear}. In level 1, we now define the large nodes as those 
whose subtree size is larger than $(1/\tau)(\log n)^3$; in level 2, larger 
than $(1/\tau)(\log\log n)^3$; and in general in level $k$ as those with 
subtree size larger than $(1/\tau)(\log^{[k]} n)^3$. This makes the space of 
all the $C_i(x)$ structures to be $o(n)$ {\em bits}. The price is that the 
traversal of the smallest trees now produces 
$O((1/\tau)(\log^{[\kappa]} n)^3)$ candidates, but this is easily sorted out 
by using $\kappa+1$ levels, since 
$(\log^{[\kappa+1]} n)^3 = o(\log^{[\kappa]} n)$.
To obtain succinct space, we will need that there are $o(n)$ subtrees of the
smallest size, but that we find only $O((1/\tau)\log^* n)$
candidates in total. Thus we set $\kappa = \log^* n - \log^{**} n$, so that
there are $O(\kappa)=O(\log^* n)$ levels, and the last-level subtrees are of 
size $O((1/\tau)(\log^{[\kappa+1]} n)^3) = 
o((1/\tau)\log^{[\log^* n - \log^{**} n]} n) = o((1/\tau)\log^* n)$.
Still, there are $O(\tau n/(\log^{[\kappa+1]} n)^3) =
O(\tau n/(\log^{[\log^* n - \log^{**} n+1]} n)^3) =
O(\tau n/(\log \log^* n)^3) = o(n)$ subtrees in the last level.

The topology of the whole tree $T$ can be represented using balanced parentheses
in $2n+o(n)$ bits, supporting in constant time all the standard tree traversal 
operations we use~\cite{NS14}. We assume that opening and closing parentheses are 
represented with 1s and 0s in $P$, respectively. Let us now focus on the less 
standard operations needed.

\subsection{Counting labels in paths}

In Section~\ref{sec:nlogn},
we count the number of times a label $\ell$ occurs in the path from $u$ to
the root by means of a query $\labelanc$ and by storing $\pathcount$ fields
in the nodes. In Section~\ref{sec:linear}, we use in addition a string $S$ to 
support range majority queries on the unary paths.

To solve $\labelanc$ queries, we 
use the representation of Durocher~\etal~\cite[Lem.~7]{DSST16}, which uses 
$nH + 2n + o(n)(H+1)$ bits in addition to the $2n+o(n)$ bits of the tree 
topology. This representation includes a string $S[1..n]$ where all the labels 
of $T$ are written in preorder; any implementation of $S$ supporting $\access$,
$\rank$, and $\select$ in time $O(\log\log_w\sigma)$ can be used (e.g., 
\cite{BN15}). This string
can also play the role of the one we call $S$ in Section~\ref{sec:linear}, 
because the labels of unary paths are contiguous in $S$, and any node $v$ can 
access its label from $S[\preorder(v)]$. 

On top of this string we must also answer range $\tau$-majority 
queries in time $O((1/\tau)\log\log_w\sigma)$. We can use the slow variant
of the succinct structure described in Section~\ref{sec:range},
which requires only $o(n)(H+1)$ additional bits and also supports $\access$ in 
$O(1)$ time and $\rank$ and $\select$ in time $O(\log\log_w\sigma)$. This
variant of the structure is built in $O(n\log n)$ time.

In addition to supporting operation $\labelanc$, we need to store or compute
the $\pathcount$ fields. Durocher~\etal~\cite{DSST16} also require this 
field, but find no succinct way to represent it. We now show a way to obtain
this value within succinct space.

%\begin{lemma}
%Let $T$ be a tree and $P[1..2n]$ be its sequence of balanced parentheses.
%Let $O(i,j)$ and $C(i,j)$ be the sets of the nodes of $T$ corresponding to 
%opening and closing parentheses, respectively, in $P[i..j]$. Then, for any
%opening parenthesis $P[i]$, it holds that $O(1,\close(P,i)) \cup 
%C(\close(P,i),2n)$ contains all the nodes of $T$, whereas $O(1,\close(P,i)) 
%\cap C(\close(P,i),2n)$ contains the ancestors of $x$ (including $x$).
%\end{lemma}
%\begin{proof}
%The set $O(1,\close(P,i))$ contains any node in the preorder traversal of $T$,
%up to completing the subtree rooted by $P[i]$. The set $C(\close(P,i),2n)$
%contains any node in the postorder traversal of $T$, from the node $P[i]$ to 
%the root of $T$. Their union includes all the nodes, but the node $P[i]$ and its
%ancestors are included in both sets.
%\end{proof}

The sequence $S$ lists the labels of $T$ in preorder, that is, aligned with
the opening parentheses of $P$. Assume we have another sequence $S'[1..n]$ 
where the labels of $T$ are listed in postorder (i.e., aligned with the closing 
parentheses of $P$). Since the opened parentheses not yet closed in $P[1..i]$
are precisely node $i$ and its ancestors, we can compute the number of times a 
label $\ell$ appears in the path from $P[i]$ to the root as
$$
 \rank_\ell(S,\rank_1(P,i)) - \rank_\ell(S',\rank_0(P,i)).
$$

Therefore, we can support this operation with $nH + o(n)(H+1)$ additional bits.
Note that, with this representation, we do not need the operation $\labelanc$,
since we do not need that $P[i]$ itself is labeled $\ell$. 

If we do use operation $\labelanc$, however, we can ensure that $P[i]$ is 
labeled $\ell$, and another solution is possible based on partial rank queries. 
Let $o=\rank_\ell(S,\rank_1(P,i))$ and $c=\rank_\ell(S',\rank_0(P,i))$ be
the numbers of opening and closing parentheses up to $P[i]$, so that we want to compute
$o-c$. Since $P[i]$ is labeled $\ell$, it holds that $S[\rank_1(P,i))]=\ell$, 
and thus $o=\prank(S,\rank_1(P,i))$. To compute $c$, we do not store $S'$, but 
rather $S''[1..2n]$, so that $S''[i]$ is the label of the node whose 
opening or closing parenthesis is at $P[i]$ (i.e., $S''$ is formed by
interleaving $S$ and $S'$). Then, $\prank(S'',i)=o+c$; therefore the answer
we seek is $o-c = 2\cdot\prank(S,\rank_1(P,i))-\prank(S'',i)$.

We use the structure for constant-time partial rank queries \cite[Sec.~3]{BN13} 
that requires $O(n)+o(nH)$ bits on top of a sequence that can be accessed
in $O(1)$ time. We can build it on $S$ and also on $S''$, though we do not 
explicitly represent $S''$: any access to $S''$ is simulated in constant time 
with $S''[i] = S[\rank_1(P,i)]$ if $P[i]=1$, and $S''[i] =
S[\rank_1(P,\open(P,i))]$ otherwise. This partial rank structure is built
in $O(n)$ randomized time and in $O(n\log n)$ time w.h.p.%
\footnote{It involves building perfect hash functions, which succeeds with
constant probability $p$ in time $O(n)$. Repeating $c\log n$ times, the
failure probability is $1-O(1/n^{c/\log(1/p)})$.}

\subsection{Other data structures}

The other fields stored at tree nodes, which we must now compute within 
succinct space, are the following:

\subparagraph*{Pointers to candidate sets $C_i(x)$}
All the branching nodes in all subtrees except those of level $\kappa+1$ are
marked, and there are $O(n/(\log^{[\kappa+1]} n)^3)=o(n)$ such nodes. We can
then mark their preorder ranks with 1s in a bitvector $M[1..n]$. Since $M$ has
$o(n)$ 1s, it can be represented within $o(n)$ bits \cite{RRR07} while 
supporting constant-time $\rank$ and $\select$ operations. We can then find out
when a node $i$ is marked (iff $M[\preorder(i)]=1$), and if it is, its rank
among all the marked nodes, $r = \rank_1(M,\preorder(i))$.
The $C_i(x)$ sets of all the marked nodes $x$ of any level can be written 
down in a contiguous memory area of total size $o(n)$ bits, sorted by the
preorder rank of $x$. A bitvector $C$ of length $o(n)$ marks the starting position
of each new node $x$ in this memory area. Then the area for marked node $i$ 
starts at $p = \select_1(C,r)$. A second bitvector $D$ can mark the starting 
position of each $C_j(x)$ in the memory area of each node $x$, and thus we access
the specific set $C_j(x)$ from position $\select_1(D,\rank_1(D,p-1)+j)$.

\subparagraph*{Pointers to subtree roots}
We store an additional bitvector $B[1..2n]$, parallel to the parentheses 
bitvector $P[1..2n]$. In $B$, we mark with 1s the positions of the opening and 
closing parentheses that are roots of subtrees of any level. As there are 
$O(n/(\log^{[\kappa+1]} n)^3)=o(n)$ such nodes, $B$ can be represented within 
$o(n)$ bits while supporting constant-time $\rank$ and $\select$ operations. 
We also store the sequence of $o(n)$ parentheses $P'$ corresponding
to those in $P$ marked with a 1 in $B$. Then the nearest subtree root
containing node $P[i]$ is obtained by finding the nearest position to the left 
marked in $B$, $r = \rank_1(B,i)$ and $j = \select_1(B,r)$, and then 
considering the corresponding node $P'[r]$. If it is an opening parenthesis, 
then the nearest subtree root is the node whose parenthesis opens in $P[j]$. 
Otherwise, it is the one opening at $P[j']$, where 
$j'=\select_1(B,\enclose(P',\open(P',r)))$ (see \cite[Sec.~4.1]{RNO11}).

\subparagraph*{Finding the nearest branching ancestor}
A unary path looks like a sequence of opening parentheses followed by a sequence
of closing parentheses. The nearest branching ancestor of $P[i]$ can then be 
obtained in constant time by finding the nearest closing parenthesis to the 
left, $l=\select_0(\rank_0(P,i))$, and the nearest opening parenthesis to the 
right, $r=\select_1(\rank_1(\close(P,i))+1)$.
Then the answer is the larger between $\enclose(P,\open(P,l))$ and 
$\enclose(P,r)$.

\subparagraph*{Determining the subtree level of a node}
Since we can compute $s=\subtreesize(i)$ of a node $P[i]$ in 
constant time, we can determine the corresponding level: 
if $s > (1/\tau)\log^3 n$, it is level 1. Otherwise, we look up $\tau 
\cdot s$ in a precomputed table of size $O(\log^3 n)$ that stores the 
level corresponding to each possible size.

\bigskip

Therefore, depending on whether we represent both $S$ and $S'$ or use partial
rank structures, we obtain two results within succinct space.

\begin{theorem} \label{thm:succ1}
Let $T$ be a tree of $n$ nodes with labels in $[1..\sigma]$, and $0<\tau<1$.
On a RAM machine of $w$-bit words, we can build in $O(n\log n)$ time a data
structure using $2nH + 4n + o(n)(H+1)$ bits, where $H \le \lg\sigma$ is the 
entropy of the distribution of the node labels,
that answers path $\tau$-majority queries in time
$O((1/\tau)\log^* n \log\log_w\sigma)$. 
\end{theorem}

\begin{theorem} \label{thm:succ2}
Let $T$ be a tree of $n$ nodes with labels in $[1..\sigma]$, and $0<\tau<1$.
On a RAM machine of $w$-bit words, we can build in $O(n\log n)$ time (w.h.p.)
a data structure using $nH + O(n) + o(nH)$ bits, where $H \le \lg\sigma$ 
is the entropy of the distribution of the node labels,
that answers path $\tau$-majority queries in time
$O((1/\tau)\log^* n \log\log_w\sigma)$. 
\end{theorem}

We note that, within this space, all the typical tree navigation functionality,
as well as access to labels, is supported.

\section{Conclusions}

We have presented the first data structures that can efficiently find the
$\tau$-majorities on the path between any two given nodes in a tree. Our
data structures use linear or near-linear space, and even succinct space,
whereas our query times are close to optimal, by a factor near log-logarithmic.

As mentioned in the Introduction, many applications of these results require that the trees are multi-labeled,
that is, each node holds several labels. We can easily accommodate multi-labeled
trees $T$ in our data structure, by building a new tree $T'$ where each node 
$u$ of $T$ with $m(u)$ labels $\ell_1,\ldots,\ell_{m(u)}$ is replaced by an
upward path of nodes $u_1,\ldots,u_{m(u)}$, each $u_i$ holding the label 
$\ell_i$ and being the only child of $u_{i+1}$ (and $u_{m(u)}$ being a child 
of $v_1$, where $v$ is the parent of 
$u$ in $T$). Path queries from $u$ to $v$ in $T$ are then transformed into path
queries from $u_1$ to $v_1$ in $T'$, except when $u$ ($v$) is an ancestor of 
$v$ ($u$), in which case we replace $u$ ($v$) by $u_{u(m)}$ ($v_{m(v)}$) in 
the query. All our complexities then hold on $T'$, which is of size $n = |T'| =
\sum_{u \in T} m(u)$.

Our query time for path $\tau$-majorities in linear space,
$O((1/\tau)\log^* n \log\log_w \sigma)$, is over the optimal time $O(1/\tau)$ 
that can be obtained for range $\tau$-majorities on sequences \cite{bgmnn2016}.
It is open whether we can obtain optimal time on trees within linear or 
near-linear space. Other interesting research problems are solving 
$\tau'$-majority queries for any $\tau' \ge \tau$ given at query time, in time 
proportional to $1/\tau'$ instead of $1/\tau$, and to support insertions and 
deletions of nodes in $T$. Similar questions can be posed for $\tau$-minorities,
where the $O((1/\tau)\log\log_w \sigma)$ query time of our linear-space solutions is also
over the time $O(1/\tau)$ achievable on sequences~\cite{bgmnn2016}.

\bibliography{paper}

\appendix

\section{Path $\tau$-Minorities}

A path $\tau$-minority query asks for a $\tau$-minority in a given path $P_{u,v}$, i.e., a label that appears at least once and at most $\tau\cdot|P_{uv}|$ times in this path. 
If we try out $A=1+\lfloor 1/\tau \rfloor$ distinct elements in the path from
$u$ to $v$, then one of them will turn out to be a $\tau$-minority. With this 
idea, we extend the technique of Chan~\etal~\cite{cdsw2015} to tree paths. To 
find a $\tau$-minority, we will find $A$ distinct labels (or all the labels,
if there are not that many) in the path $P_{uz}$, where $z=\lca(u,v)$, and check
their frequency in $P_{uv}$. We then run an analogous process on the path
$P_{vz}$. We will stop as soon as we find a label that is not a
$\tau$-majority. We describe the process on $P_{uz}$, as $P_{vz}$ is
analogous.

To find $A$ distinct labels, we will simulate on $P_{uz}$ the algorithm of 
Muthukrishnan~\cite{Mut02}, which finds $A$ distinct elements in any range 
of an array $E$. In his algorithm, Muthukrishnan defines the array $C$ where $C[i] = 
\max \{ j < i,~ E[j] = E[i] \}$ ($C[i]$ is set to $0$ if such a value does not exist) and builds on $C$ a 
range minimum query (RMQ) data structure; a range minimum query asks for the minimum element in a given subrange of the array.
Then he finds $A$ (or all the)
distinct elements in any range $E[i..j]$ via $O(A)$ RMQs.

In our case, we store for each node $u$ the field
$\prevlabel(u) = \nodedepth(\labelanc(\parent(u), \linebreak \nodelabel(u)))$, which is the depth of the 
nearest ancestor of $u$ with its same label (and $-1$ if there is none).
Then we conceptually define $E$ and $C$ over $P_{u,v}$, where $E[i]=\nodelabel(\anc(u,i+\nodedepth(z)-1))$ and
$C[i]=1+\prevlabel(\anc(u,i+\nodedepth(z)-1))$.
Note that we do not store $E$ or $C$ explicitly, but each entry of $E$ or $C$ can be computed in constant time using these formulas. 
To solve RMQs on $C$, we also build the linear-space data structure of Chazelle \cite{Cha87},
which can return the minimum-weight node in any path of a weighted tree in constant time. This data structure is constructed over the tree $T$, for which we assign $\prevlabel(u)$ as the weight of each node $u$. 
%Then
With all these structures, 
we can run Muthukrishnan's
algorithm and obtain the $A$ distinct labels of $P_{uz}$.
%To solve RMQs on $P_{uz}$, we use the data structure of Chazelle \cite{Cha87},
%which takes constant time and linear space.
This yields our first result, which
slightly reduces the $O((1/\tau)\log\log n)$ time (within linear space) of
Durocher~\etal~\cite{DSST16}. Note that the $\prevlabel$ fields are easily
computed in $O(n)$ time in a DFS traversal.

\begin{theorem} \label{thm:linear-min}
Let $T$ be a tree of $n$ nodes with labels in $[1..\sigma]$, and $0<\tau<1$.
On a RAM machine of $w$-bit words, we can build an $O(n)$ space data
structure that answers path $\tau$-minority queries in time
$O((1/\tau)\log\log_w\sigma)$. The structure is built in linear time.
\end{theorem}

It is likely that the result of Durocher~\etal~\cite{DSST16} can be improved
to match ours, by just using a faster predecessor data structure. We can,
however, make our solution succinct by using our tree representation of
$2n+o(n)$ bits \cite{NS14}. Instead of storing field $\prevlabel$, we compute
it on the fly with the given formula. Using the structures of
Durocher~\etal~\cite[Lem.~7]{DSST16}, we can compute $\labelanc$ in 
time $O(\log\log_w\sigma)$.
% GZL: the next part is almost true, but select requires omega(1), which will
% mess things up, e.g., alpha^2(n)+..., and give randomized construction time.
%constant 
%time, because the label we look for is the same label of $u$, so $u$ can be 
%mapped to the tree $T_{\nodelabel(u)}$, take its parent in that tree, and 
%return to the corresponding node in $T$. The mapping between trees can be
%implemented in constant time using $\prank$ queries on $S$. 
Their structure uses $2n+o(n)$ bits
in addition to the topology of $T$ and the representation of $S$.

The structure for RMQs, on the other hand, can be replaced by the one of 
Chan~\etal~\cite{CHMZ17}, which uses $2n+o(n)$ further bits and answers RMQs
with $O(\alpha(n))$ queries $\prevlabel(u)$, where $\alpha$ is the inverse 
Ackermann function. Therefore, we can spot the $A$ candidates in time 
$O(A \cdot \alpha(n) \log\log_w\sigma)$ and then verify them in time 
$O(A \cdot \log\log_w\sigma)$. This yields the first result for path
$\alpha$-minority queries within succinct space.

%\begin{theorem} \label{thm:succinct1-min}
%Let $T$ be a tree of $n$ nodes with labels in $[1..\sigma]$, and $0<\tau<1$.
%On a RAM machine of $w$-bit words, we can build in $O(n)$ randomized time a data
%structure using $nH + O(n) + o(nH)$ bits, where $H \le \lg\sigma$ is the 
%entropy of the distribution of the node labels,
%that answers path $\tau$-minority queries in time
%$O((1/\tau)(\alpha(n) + \log\log_w\sigma))$, where $\alpha$ is the inverse
%Ackermann function.
%\end{theorem}
%
%By replacing the $\prank$ queries with general $\rank$ queries, we obtain
%linear deterministic construction time.

\begin{theorem} \label{thm:succinct2-min}
Let $T$ be a tree of $n$ nodes with labels in $[1..\sigma]$, and $0<\tau<1$.
On a RAM machine of $w$-bit words, we can build in $O(n)$ time a 
data structure using $nH + 6n + o(n)(H+1)$ bits, where $H \le \lg\sigma$ is the 
entropy of the distribution of the node labels,
that answers path $\tau$-minority queries in time
$O((1/\tau)\alpha(n) \log\log_w\sigma)$, where $\alpha$ is the inverse
Ackermann function.
\end{theorem}

\end{document}